\documentclass[11pt]{article}
\usepackage{latexsym}
\usepackage{amssymb}
\usepackage{amsmath,amsthm}
\usepackage[numbers,sort&compress]{natbib}
\usepackage{color}
\usepackage{amsfonts,amssymb}
\usepackage{mathrsfs}
\usepackage{dsfont}
\usepackage{subfigure}
\usepackage{graphicx}
\usepackage[colorlinks,linkcolor=blue,citecolor=blue]{hyperref}

\newtheorem{thm}{Theorem}[section]
\newtheorem{lem}{Lemma}[section]

\numberwithin{equation}{section}
\def \h#1{\widehat{#1}}
\def \t#1{\widetilde{#1}}

\newcommand{\bse}{\begin{subequations}}
\newcommand{\ese}{\end{subequations}}

\def \b#1{\bar{#1}}

\def \llbr{[\![}
\def \rrbr{]\!]}

\newcommand{\be}{\begin{equation}}
\newcommand{\ee}{\end{equation}}

\begin{document}
\title{On a Second Discretization of the ZS-AKNS Spectral Problem: Revisit}
\author{Kui Chen,~ Xiao Deng,~ Da-jun Zhang\footnote{Corresponding author: djzhang@staff.shu.edu.cn} \\
Department of Mathematics, Shanghai University, Shanghai 200444, P.R. China}
\date{\today}
\maketitle

\begin{abstract}
In this paper we revisit a discrete spectral problem which was proposed by Ragnisco and Tu in 1989, as
a second discretization of the ZS-AKNS spectral problem.
We show that the spectral problem corresponds to a bidirectional discretization of the derivative of two
wave functions $\phi_{1,x}$ and $\phi_{2,x}$.
As a connection with higher dimensional systems,
the spectral problem and a related hierarchy can be derived from Lax triads of the differential-difference KP hierarchy
via a symmetry constraint.
Isospectral and nonisospectral flows derived from the spectral problem compose a Lie algebra.
By considering its infinite dimensional subalgebras  and continuum limit of recursion operator,
three semi-discrete AKNS hierarchies are constructed.

\vskip 6pt

\noindent
{\bf Keywords:} ZS-AKNS spectral problem, gauge transformation, symmetry constraint, differential-difference KP equation

\noindent
{\bf MSC (2010):}  35Q51, 37K60
\end{abstract}

\section{Introduction} \label{Sec-1}

It is well known that the fundamental ZS-AKNS spectral problem\cite{ZakS-JETP-1972,AKNS-PRL-1973}
\begin{equation}\label{sp-AKNS}
\Phi_x = \left(
                    \begin{array}{cc}
                      \eta & q \\
                      r & -\eta \\
                    \end{array}
                  \right) \Phi,~~~
\Phi=(\phi_1, \phi_2)^T
\end{equation}
has a discretization given by Ablowitz and Ladik \cite{AL-JMP-1975,AL-JMP-1976}:
\begin{equation}\label{sp-AL}
\Phi_{n+1} = \left(
                    \begin{array}{cc}
                      \lambda & Q_n \\
                      R_n & 1/\lambda \\
                    \end{array}
                  \right) \Phi_n,~~~
\Phi_n=(\phi_{1,n}, \phi_{2,n})^T,
\end{equation}
which bears their names and is called the Ablowitz-Ladik (AL) spectral problem.
Here $f_n$ stands for a function $f(n,t)$ defined on $\mathbb{Z}\times \mathbb{C}$.
\eqref{sp-AL} leads to a semidiscrete AKNS (sdAKNS) hierarchy through suitably combining the AL flows,
and one-field reductions yield the semidiscrete
KdV, modified KdV and nonlinear Schr\"odinger hierarchies (cf.\cite{ZhaC-SAM-2010-I,ZhaC-SAM-2010-II,FQSZ-arxiv-2013}).
In 1989 Ragnisco and Tu proposed a discrete spectral problem\cite{RT-prin-1989}
\begin{equation}\label{sp-DT}
\Theta_{n+1} = \left(
                    \begin{array}{cc}
                      \lambda^2 + Q_n R_{n} & Q_n \\
                      R_{n} & 1
                    \end{array}
                  \right) \Theta_n,~~~
\Theta_n=(\theta_{1,n}, \theta_{2,n})^T,
\end{equation}
which was then studied in \cite{ZTOF-JMP-1991,MRT-IP-1994}.
\eqref{sp-DT} leads to a hierarchy of semidiscrete equations
which recover the AKNS hierarchy in continuum limit but one-field reduction was not available \cite{MRT-IP-1994}.

In this paper, we revisit the discrete spectral problem \eqref{sp-DT}.
It will be shown that \eqref{sp-DT} is gauge equivalent to the form
\begin{equation}\label{sp-1}
 \left( \begin{array}{c}
            \phi_{1,n+1} \\
            \phi_{2,n-1} \\
          \end{array}
        \right) = \left(
                    \begin{array}{cc}
                      \lambda & Q_n \\
                       -R_n & \lambda \\
                    \end{array}
                  \right)  \left(
                             \begin{array}{c}
                             \phi_{1,n} \\
                             \phi_{2,n} \\
                             \end{array}
                           \right).
\end{equation}
Compared with the AL spectral problem \eqref{sp-AL},
the above one is obtained by discretizing the first order derivatives of wave functions $(\phi_1,\phi_2)^T$ in \eqref{sp-AKNS} in bidirection, i.e.
\begin{equation}
\phi_{1,x}\sim \frac{\phi_{1,n+1}-\phi_{1,n}}{\epsilon},~~ \phi_{2,x}\sim  \frac{\phi_{2,n}-\phi_{1,n-1}}{\epsilon}.
\end{equation}
Both \eqref{sp-AL} and \eqref{sp-DT} recover the AKNS spectral problem in continuum limit
by defining
\[\Phi(n+j)=\Phi(x+j\epsilon),~~ (Q_n, R_n)=\epsilon (q,r),~~~ \lambda=e^{\epsilon \eta},\]
and taking $\epsilon \to 0$.

Besides bidirectional discretization of \eqref{sp-AKNS},
the discrete spectral problem \eqref{sp-DT} or \eqref{sp-1}
is interesting in two  more aspects.
One is that \eqref{sp-DT} is related to a symmetry constraint
of the differential-difference Kadomtsev-Petviashvili (KP) equation.
This fact demonstrates a link between (1+1)-dimensional and
(2+1)-dimensional semidiscrete integrable systems.
The other is that, as a spectral problem, \eqref{sp-DT} is  a Darboux transformation of the AKNS hierarchy.
Note that a Darboux transformation can act as a discrete spectral problem to generate
semidiscrete and fully discrete integrable systems.
Let us give more details in the following.

It is well known that in continuous case the AKNS hierarchy can be viewed as a symmetry constraint of the Lax pairs of the KP
hierarchy \cite{KSS-PLA-1991,KS-IP-1991,CL-PLA-1991,CL-JPA-1992}.
The differential-difference KP (D$^2\Delta$KP) equation\footnote{D$^2\Delta$ indicates 2 continuous and 1 discrete independent variables.}
reads \cite{DJM-JPSJ-1982}
\begin{equation}
\Delta\left(\frac{\partial u}{\partial t_2}+2\frac{\partial u}{\partial x}-2u\frac{\partial u}{\partial x}\right)
=(2+\Delta)\frac{\partial^2 u}{{\partial x}^2},
\label{DDKP}
\end{equation}
which is related to the spectral problem \cite{KT-SCF-1997}
\begin{equation}
\mathfrak{L}\varphi_n =\xi \varphi_n,~~ \mathfrak{L}=\Delta+u_{0,n}+ u_{1,n}\Delta^{-1}+ u_{2,n}\Delta^{-2}+\cdots,
\label{sp-L}
\end{equation}
where $\Delta =E-1$, $Ef_n=f_{n+1}$, $\mathfrak{L}$ is called a pseudo-difference operator and in \eqref{DDKP} $u=u_{0,n}$.
In this paper we will show that by a symmetry constraint the spectral problem \eqref{sp-L} leads to a
spectral problem which is gauge-equivalent to \eqref{sp-DT}
and the Lax triads of the D$^2\Delta$KP hierarchy yields a sdAKNS hierarchy. This link
will be explained in detail in Sec.\ref{Sec-3}.

It is also well known that a Darboux transformation $\t \Phi = D(u, \t u, \lambda)\Phi$
of a continuous spectral problem $\Phi_x=M(u,\eta)\Phi$, where
$D(u, \t u, \lambda)$ is a Darboux matrix with parameter $\lambda$ and $\t \Phi$ and $\t u$ stand for new eigenfunction and potential
corresponding to $\lambda$, can act as a discrete spectral problem (by considering $\t \Phi=\Phi_{n+1}$ and $\t u=u_{n+1}$)
\begin{equation}
\Phi_{n+1} = D(u_n, u_{n+1}, \lambda)\Phi_n
\end{equation}
to generate semidiscrete integrable systems as a compatible condition with $\Phi_x=M(u,\eta)\Phi$ \cite{LB-PNAS-1980,Levi-JPA-1981}.
Moreover, Darboux transformations with different parameters, say
\[\t \Phi = D(u, \t u, \lambda_1)\Phi,~~\h \Phi = D(u, \h u, \lambda_2)\Phi\]
can be used as a Lax pair to generate fully discrete integrable systems.
As examples one can refer to \cite{CZ-CPL-2012,CaoZ-JPA-2012,KMW-TMP-2013,Mik-INI-2013}.
In \cite{AY-JPA-1994} the spectral problem \eqref{sp-DT} was studied as a Darboux transformation
of the ZS-AKNS spectral problem \eqref{sp-AKNS} (with $\lambda^2=2(\eta-\gamma)$).
It is natural that the semidiscrete equations generated from a Darboux transformation (as a discrete spectral problem)
are related via suitable continuum limit to the original continuous spectral problem.

In this paper, as new results we mainly achieve the following:
\begin{itemize}
\item{
find that the discrete spectral problem \eqref{sp-DT}  is (gauge) equivalent to \eqref{sp-1}
which is a bidirectional discretization of the ZS-AKNS spectral problem \eqref{sp-AKNS};}
\item{build  connection between the spectral problem \eqref{sp-DT} and   pseudo-difference operator spectral problem \eqref{sp-L}
via a  symmetry constraint of the D$^2\Delta$KP equation, as well as a connection between a sdAKNS hierarchy and the Lax triads of the
D$^2\Delta$KP hierarchy;}
\item{
obtain three sdAKNS hierarchies that are different from those derived from the AL spectral problem \eqref{sp-AL}.}
%\item{
%prove nonlocal one-field reductions of these (isospectral) sdAKNS hierarchies and obtain  semidiscrete scalar nonlinear integrable equations.}
\end{itemize}

The paper is organized as follows.
Sec.\ref{Sec-2} contains necessary notions and notations.
In Sec.\ref{Sec-3} we show connections between \eqref{sp-DT}, \eqref{sp-1} and \eqref{sp-L},
and in Sec.\ref{Sec-4} we derive \eqref{sp-DT} and a sdAKNS hierarchy from the Lax triads of the
D$^2\Delta$KP hierarchy.
In Sec.\ref{Sec-5} we discuss possible sdAKNS hierarchies related to \eqref{sp-DT}.
Sec.\ref{Sec-6} is for conclusions and discussions.
There is one Appendix which, as a comparison, gives the  sdAKNS hierarchies derived from the AL spectral problem.

\section{Basic notions}\label{Sec-2}

Let us shortly describe some notions and notations that we will use in the paper.
(We mainly follow  \cite{FF-PD-1981,Fuc-MSE-1991}).

For  functions $Q_n$ and $R_n$ defined on $\mathbb{Z}$ and vanishing  rapidly as $n \to \pm \infty$,
let $U_n \doteq (Q_n,R_n)^T$.
Consider a differential-difference evolution equation
\begin{align}\label{def:nlee}
U_{n,t}=K(U_n),\quad U_n\in \mathcal{M},
\end{align}
where by $\mathcal{M}$ we denote the infinite dimensional linear manifold of  functions $U_n$.
The solution $U_n=U(n,t)$ is usually asked to depend in a $\mathbb{C}^{\infty}$-way on the time parameter $t$.
Let $S$ be  the fiber of the tangent bundle $T\mathcal{M}$ at any point $U_n \in M$.
In principle there is an identification between the linear spaces $\mathcal{M}$ and $S$,
but it is convenient to regard them as different objects for a better geometrical
understanding (i.e. $\mathcal{M}$ is the manifold under examination, $S$ is the tangent space at any point $U_n \in \mathcal{M}$).
Let $S^*$  be the dual space of $S$ w.r.t. the bilinear form $\langle\cdot,\cdot\rangle: S^*\times S\rightarrow\mathbb{R}$ defined as
\begin{align}\label{bil}
\langle f_n, g_n\rangle= \sum_{n=-\infty}^{+\infty} f_n g_n,~~~ f_n\in S^*, g_n\in S.
\end{align}

The G\^ateaux derivative of a function (or an operator or a functional) $F(U_n)$ on $\mathcal{M}$ in the direction $g_n\in S$ is defined as
\begin{align}\label{def:G-deriv}
F'[g_n]=\frac{\partial}{\partial\varepsilon}F(U_n+\varepsilon g_n)\Big|_{\varepsilon=0},\quad U_n\in \mathcal{M},g_n\in S.
\end{align}
The above definition is valid as well for the case $F=F(U_n,t)$,
where $F$ depends explicitly on the time parameter $t$
and we treat $U_n$ and $t$  as independent variables (cf.  \cite{Fuc-MSE-1991}).
For the sake of a more generic sense, in the following definitions are given for the time-dependent cases.
They are  valid as well when we remove the independent time variable $t$.

For two vector fields $F(U_n,t), G(U_n,t) : \mathcal{M} \times \mathbb{R} \to S$, their standard commutator is defined as
\begin{align}\label{def:comm}
\llbr F,G\rrbr=F'[G]-G'[F].
\end{align}
Vector field $G(U_n,t) : \mathcal{M} \times \mathbb{R} \to S$ is called a symmetry of equation \eqref{def:nlee} if
\begin{align}\label{def:sym}
\partial_t G(U_n,t)+\llbr G(U_n,t),K(U_n)\rrbr=0
\end{align}
holds everywhere in $\mathcal{M}\times \mathbb{R}$.

A linear operator $L(U_n,t): S\rightarrow S$ is called a strong symmetry operator of equation \eqref{def:nlee} if
\begin{align}\label{def:SS-op}
\partial_t L +L'[K]=[K', L]
\end{align}
holds everywhere on $\mathcal{M}$, where $[A,B]=AB-BA$.
A linear operator $L(U_n,t): S\rightarrow S$ is called to be hereditary (or a hereditary operator) if
\begin{align}\label{def:Here-op}
L'[LF]G - L'[LG]F = L(L'[F]G-L'[G]F),~~ \forall F,G\in S.
\end{align}
If $L$ is a hereditary operator, so is $L^{-1}$.
If $L$ is a  hereditary operator and is a strong symmetry of equation \eqref{def:nlee},
then $L$ is also a strong symmetry of equation $U_{n,t}=L\,K(U_n)$.

\section{Gauge equivalent forms of \eqref{sp-DT}}\label{Sec-3}

In this section we will  list out some spectral problems which are  gauge equivalent to the spectral problem \eqref{sp-DT}.

In addition to \eqref{sp-DT} and \eqref{sp-1}, we list the following spectral problems
\begin{equation}\label{sp-1*}
\Phi_{n+1} = \left(
                    \begin{array}{cc}
                      \lambda & Q_n \\
                      R_{n+1} & (1+Q_n R_{n+1})/\lambda \\
                    \end{array}
                  \right) \Phi_n,~~~
\Phi_n=(\phi_{1,n}, \phi_{2,n})^T,
\end{equation}
\begin{equation}\label{sp-2*}
\Psi_{n+1} = \left(
                    \begin{array}{cc}
                      \lambda^2 & Q_n \\
                      \lambda^2 R_{n+1} & 1+Q_n R_{n+1} \\
                    \end{array}
                  \right) \Psi_n,~~~
\Psi_n=(\psi_{1,n}, \psi_{2,n})^T,
\end{equation}
\begin{equation}\label{sp-2**}
\Sigma_{n+1} = \left(
                    \begin{array}{cc}
                      1 & Q_n \\
                      R_{n+1}/\lambda^2 & (1+Q_n R_{n+1})/\lambda^2 \\
                    \end{array}
                  \right) \Sigma_n,~~~
\Sigma_n=(\sigma_{1,n}, \sigma_{2,n})^T,
\end{equation}
\begin{equation}\label{sp-2***}
\Pi_{n+1} = \left(
                    \begin{array}{cc}
                      \lambda^2 & \lambda Q_n \\
                      \lambda R_{n+1} & 1+Q_n R_{n+1} \\
                    \end{array}
                  \right) \Pi_n,~~~
\Pi_n=(\pi_{1,n}, \pi_{2,n})^T,
\end{equation}
\begin{equation}\label{sp-2}
 \left( \begin{array}{c}
            \psi_{1,n+1} \\
            \psi_{2,n-1} \\
          \end{array}
        \right) = \left(
                    \begin{array}{cc}
                      \lambda^2 & Q_n \\
                       -R_n & 1 \\
                    \end{array}
                  \right)  \left(
                             \begin{array}{c}
                             \psi_{1,n} \\
                             \psi_{2,n} \\
                             \end{array}
                           \right),
\end{equation}
\begin{equation}
\mathcal{L} \varphi_n=\xi \varphi_n,~~~ \mathcal{L}=\Delta- Q_nR_n- Q_n\Delta^{-1} R_n.
\label{sp-SC}
\end{equation}
They are related to each other as in the following diagram.

\[\begin{array}{cccccl}
\eqref{sp-1} & \underrightarrow{~~~ \mathrm{GT}_1~~~ } & \eqref{sp-2} & \underrightarrow{~~ \varphi_n= \psi_{1,n},  \xi=\lambda^2-1~} & \eqref{sp-SC}& \\
\Big\Updownarrow &                      & \Big\Updownarrow & & \Big\uparrow & \hskip -10pt  \varphi_n= \theta_{1,n},~ \xi=\lambda^2-1\\
\eqref{sp-1*} & \underrightarrow{~~~ \mathrm{GT}_1~~~ } & \eqref{sp-2*} & \underrightarrow{~~~~~~~~~~~~~ \mathrm{GT}_2 ~~~~~~~~~~~~} & \eqref{sp-DT}& \\
\Big\downarrow & \hskip -50pt \mathrm{GT}_3          & \Big\downarrow & \hskip -120pt \mathrm{GT}_4 &&\\
\eqref{sp-2**}&& \eqref{sp-2***} &&&\\
\end{array}
\]
\centerline{Fig.1 Relations of eight spectral problems}

\vskip 12pt
\noindent
Here are the gauge transformations
\begin{subequations}
\begin{align}
&\mathrm{GT}_1:~~ \Phi_n=T_1 \Psi_n,~~  T_1 = \lambda^{-n} \left(\begin{array}{cc}
                      1 & 0 \\
                      0 & 1/\lambda
                    \end{array}\right),\label{GT-T1}\\
&\mathrm{GT}_2:~~ \Psi_n=T_2 \Theta_n,~~  T_2 =  \left(\begin{array}{cc}
                      1 & 0 \\
                      R_n & 1
                    \end{array}\right),  \label{GT-T2}\\
&\mathrm{GT}_3:~~ \Phi_n=T_3 \Sigma_n,~~  T_3 =  T^{-1}_1,  \label{GT-T3}\\
&\mathrm{GT}_4:~~ \Psi_n=T_4 \Pi_n,~~  T_4 =  \left(\begin{array}{cc}
                      1/\lambda & 0 \\
                      0 & 1
                    \end{array}\right). \label{GT-T4}
\end{align}
\end{subequations}

\begin{thm}\label{Thm-0}
The spectral problem \eqref{sp-DT} is (gauge) equivalent to \eqref{sp-1} which is a bidirectional discretization
of the ZS-AKNS spectral problem \eqref{sp-AKNS}.
\end{thm}

Note that after the early work \cite{ZTOF-JMP-1991,MRT-IP-1994}, the spectral problem \eqref{sp-DT}
has been  reinvestigated
%as a new spectral problem
in different forms
(for example, \eqref{sp-1*} in \cite{FanY-IJTP-2009,TZZZ-MPLB-2016}, \eqref{sp-2*} in \cite{DonZZ-CNSNS-2016},
\eqref{sp-2**} in  \cite{YXD-CSF-2005}, \eqref{sp-2***} in \cite{DXYD-MPLB-2007,Qin-JMP-2008,WRH-DDNS-2012}), \eqref{sp-SC} in \cite{LLC-CTP-2014}, etc.).
However, since in these gauge transformations GT$_i$ only  eigenfunctions are involved (without any changes of potentials),
the evolution equations derived from all the spectral problems listed in Fig.1 are the same (up to some combinations of flows).
In fact, for two evolution equations which are derived respectively as compatibilities of linear problems
\begin{equation}\label{sp-t1}
 \Phi_{n+1}=M_n(u_n,\lambda)\Phi_n,~~ \Phi_{n,t} = N_n \Phi_n,
\end{equation}
and
\begin{equation}\label{sp-t2}
 \Psi_{n+1}=U_n(u_n,\lambda)\Psi_n,~~ \Psi_{n,t} = V_n \Psi_n,
\end{equation}
if they are gauge equivalent via transformation $\Phi_n=T_n\Psi_n$,
then there are relations
\[
 M_n = T_{n+1} {U}_n T_n^{-1},~~
 N_n = T_{n,t}T^{-1}_n + T_n V_n T^{-1}_n,
\]
and consequently their compatibilities are related by
\[
 M_{n,t}- N_{n+1}M_n+ M_n N_n
 =T_{n+1}(U_{n,t_m} - V_{n+1}U_n + U_n V_n)T_{n}^{-1},
\]
which means the two equations derived from \eqref{sp-t1} and \eqref{sp-t2}  as compatibilities are same.

\section{A symmetry constraint of the D$^2\Delta$KP equation}\label{Sec-4}

In this section we investigate in detail a symmetry constraint of the D$^2\Delta$KP equation \eqref{DDKP},
by which we reduce \eqref{sp-L} to \eqref{sp-SC}
and generate a sdAKNS hierarchy as well.

\subsection{Spectral problem \eqref{sp-SC} as a symmetry constraint of  \eqref{sp-L}}\label{sec-4-1}

Spectral problem \eqref{sp-SC} is connected with \eqref{sp-L} through a symmetry constraint of the D$^2\Delta$KP equation \eqref{DDKP}.
To explain this,  let us start from Lax triad of the D$^2\Delta$KP equation \cite{FHTZ-Non-2013}.
Consider the pseudo-difference operator $\mathfrak{L}$ defined in \eqref{sp-L}, i.e.,
\begin{equation}
\mathfrak{L}=\Delta+u_{0,n}+ u_{1,n}\Delta^{-1}+ u_{2,n}\Delta^{-2}+\cdots,
\label{L}
\end{equation}
where $u_{i,n}=u_i(n,x,\mathbf{t})$ and $\mathbf{t}=(t_1,t_2,\cdots)$. The difference operator $\Delta$  obeys
the discrete Leibniz rule
\begin{equation}
\Delta^s g(n)=\sum^{\infty}_{i=0}\mathrm{C}_s^i (\Delta^i g(n+s-i))\Delta^{s-i},~~ s\in \mathbb{Z},
\label{Leib-rule}
\end{equation}
where
\begin{equation}
\mathrm{C}_s^i=\frac{s(s-1)(s-2)\cdots(s-i+1)}{i!},~~\mathrm{C}_0^0=1.
\label{Csi}
\end{equation}
The one-field D$^2\Delta$KP hierarchy can be derived from the following Lax triad \cite{FHTZ-Non-2013},
\begin{subequations}\label{Lax-t}
\begin{align}
& \mathfrak{L} \varphi_n=\xi \varphi_n,\label{Lax-ta}\\
& \varphi_{n,x}=A_1 \varphi_n,~~ A_1=\Delta +u_{0,n},\label{Lax-tb}\\
& \varphi_{n,t_j}=A_j \varphi_n,~~ (j=1,2,\cdots),\label{Lax-tc}
\end{align}
\end{subequations}
where $A_j=(\mathfrak{L}^j)_{+}$ denotes the difference part of $\mathfrak{L}^j$,  the first two of which are
\begin{subequations}\label{A1A2}
\begin{align}
& A_1=\Delta +u_{0,n},\label{A1}\\
& A_2=\Delta^2+((\Delta u_{0,n})+2u_{0,n})\Delta+(\Delta u_{0,n})+u^2_{0,n}+(\Delta u_{1,n})+2u_{1,n}. \label{A2}
\end{align}
\end{subequations}
Compatibility of \eqref{Lax-t} reads
\begin{subequations}\label{Lax-com}
\begin{align}
& \mathfrak{L}_x = [A_1, \mathfrak{L}],\label{Lax-ca}\\
& \mathfrak{L}_{t_j} = [A_j, \mathfrak{L}],~~(j=1,2,\cdots),\label{Lax-cb}\\
& A_{1,t_j}-A_{j,x}+[A_1,A_j]=0,~~ (j=1,2,\cdots),\label{Lax-cc}
\end{align}
\end{subequations}
where $[A,B]=AB-BA$.
Among \eqref{Lax-com}, the first equation \eqref{Lax-ca} provides expressions of $u_{j,n}$ in terms of $u_{0,n}$, which are
\begin{subequations}
\begin{align}
& \Delta u_{1,n} = u_{0,n,x},   \\
& \Delta u_{k+1,n} = u_{k,n,x} - \Delta u_{k,n} -u_{0,n}u_{k,n} + \sum^{k-1}_{j=0}(-1)^j \mathrm{C}_{k-1}^{j} u_{k-j,n}
\Delta^{j} u_{0,n-k},\quad k\geq 1.
\end{align}\label{uk-u0}
\end{subequations}
Besides, \eqref{Lax-cc}, written as
\begin{equation}
u_{0,n,t_j}= \mathcal{K}_{j}=A_{j,x}-[A_1,A_j],~~ (j=1,2,\cdots),
\label{DDKP-hie}
\end{equation}
provides zero curvature representations of a hierarchy of the one-field
D$^2\Delta$KP equation if substituting $u_{j,n}$ with $u_{0,n}$ by using \eqref{uk-u0}.
Particularly, when $j=2$ one gets the D$^2\Delta$KP equation \eqref{DDKP}.

As in continuous case, the D$^2\Delta$KP equation \eqref{DDKP} has a symmetry
composed of eigenfunction $\varphi_n$ and its adjoint function $\b\varphi_n$.
\begin{lem}\label{Lem-1}\cite{KT-SCF-1997}
$(\varphi_n \b\varphi_n)_x$ is a symmetry of the D$^2\Delta$KP equation \eqref{DDKP}
provided
\begin{equation}
\varphi_{n,x}=A_1 \varphi_n,~~ \b\varphi_{n,x}=-A_1^* \b\varphi_n,
\label{phi-x}
\end{equation}
and
\begin{equation}
\varphi_{n,t_2}=A_2 \varphi_n,~~ \b \varphi_{n,t_2}=-A^*_2 \b \varphi_n.
\end{equation}
\end{lem}

Here  $A_i^*$ stands for the formal adjoint operator of $A_i$ w.r.t. the bilinear form \eqref{bil}. Explicit forms of  them are
\begin{align*}
& A_1^*=-\Delta E^{-1} +u_{0,n}, \\
& A_2^*=\Delta^2E^{-2} -\Delta E^{-1}((\Delta u_{0,n})+2u_{0,n})+(\Delta u_{0,n})+u^2_{0,n}+ u_{0,n,x}+2(\Delta^{-1} u_{0,n,x}),
\end{align*}
where we have replaced $u_{1,n}$ with $\Delta^{-1} u_{0,n,x}$.

Note that $u_{0,n,x}$ is also a symmetry of the  D$^2\Delta$KP equation \eqref{DDKP} (cf.\cite{FHTZ-Non-2013}).
Consider a symmetry $\sigma=u_{0,n,x}+ (\varphi_n \b\varphi_n)_x$. Taking $\sigma=0$ leads to a group invariant solution to \eqref{DDKP}.
On the other hand, $\sigma=0$ provides a symmetry constraint $u_{0,n}=-\varphi_n \b\varphi_n$.
For convenience, we write $\varphi_n =Q_n,~ \b\varphi_n=R_n$, and then we have
\begin{equation}
u_{0,n}=-Q_nR_n
\label{u0-SC}
\end{equation}
and \eqref{phi-x} reads
\begin{equation}
Q_{n,x}=Q_{n+1}-Q_n-Q_n^2R_n,~~
R_{n,x}=R_n-R_{n-1}+Q_nR_n^2.
\label{QRx-SC}
\end{equation}
If replacing $R_n$ with $R_{n+1}$, \eqref{QRx-SC} provides a B\"acklund transformation for the nonlinear Schr\"odinger equations (cf.\cite{AY-JPA-1994}),
and through suitable continuum limit it yields the nonlinear Schr\"odinger equations (cf.\cite{CaoZX-prn-2016}).
Next we investigate the change of $\mathfrak{L}$ under constraint \eqref{u0-SC} and \eqref{QRx-SC}.
\begin{lem}\label{Lem-2}
If $u_{0,n}$ is given by \eqref{u0-SC} where $Q_n$ and $R_n$ obey \eqref{QRx-SC},
then $u_{k,n}$ defined in \eqref{uk-u0} can be expressed as
\begin{equation}\label{uk-QR}
u_{k+1,n}=(-1)^{k+1} Q_n \Delta^{k}R_{n-k-1},~~ k=0,1,2,\cdots.
\end{equation}
\end{lem}

\begin{proof}
From \eqref{uk-u0}, \eqref{u0-SC} and \eqref{QRx-SC} it is not difficult to find
\begin{equation*}
u_{1,n}=-Q_n R_{n-1}, ~~ u_{2,n}= Q_n \Delta R_{n-2},~~ u_{3,n}= -Q_n\Delta^{2} R_{n-3}.
\end{equation*}
Now we suppose \eqref{uk-QR} is valid up to $u_{k+1,n}$.
Then for $u_{k+2,n}$, from \eqref{uk-u0} we have
\begin{equation}
 \Delta u_{k+2,n} = u_{k+1,n,x} - \Delta u_{k+1,n} -u_{0,n}u_{k+1,n} + \sum^{k}_{j=0}(-1)^j \mathrm{C}_{k}^{j} u_{k+1-j,n}
\Delta^{j} u_{0,n-k-1}.
\label{uk+2}
\end{equation}
For the first three terms on the right hand side, substituting \eqref{uk-QR} into them and making use of \eqref{u0-SC} and \eqref{QRx-SC}, we find
\begin{align*}
&u_{k+1,n,x} - \Delta u_{k+1,n} -u_{0,n}u_{k+1,n} \\
=&\Delta((-1)^{k+2}Q_n\Delta^{k+1}R_{n-k-2})+(-1)^{k+1}Q_n\Delta^{k}(Q_{n-k-1}R^2_{n-k-1}).
\end{align*}
Using formula \eqref{Leib-rule}, the last term on the right hand side of the above equation yields
\begin{align*}
 & (-1)^{k+1}Q_n\Delta^{k}(Q_{n-k-1}R^2_{n-k-1})\\
= & (-1)^{k+2}Q_n\Delta^{k}(R_{n-k-1}u_{0,n-k-1})\\
=& (-1)^{k+2}\sum^{k}_{j=0} \mathrm{C}^j_k \,Q_n(\Delta^{k-j}R_{n-k-j-1})(\Delta^j u_{0,n-k-1})\\
=& -\sum^{k}_{j=0} (-1)^j \mathrm{C}^j_k \,u_{k+1-j,n}\Delta^j u_{0,n-k-1},
\end{align*}
which is just canceled by the last term on the right hand side of \eqref{uk+2}.
Thus we reach
\[\Delta u_{k+2,n} =\Delta((-1)^{k+2}Q_n\Delta^{k+1}R_{n-k-2}),\]
i.e.,
\[u_{k+2,n} =(-1)^{k+2}Q_n\Delta^{k+1}R_{n-k-2}.\]
Based on  mathematical induction, we complete the proof.
\end{proof}

With Lemma \ref{Lem-2} in hand and making use of formula \eqref{Leib-rule} (for $s=-1$),
we immediately find
\begin{equation}
(\mathfrak{L})_{-}=\sum^{\infty}_{j=1}u_{j,n}\Delta^{-j}=-Q_n\Delta^{-1}R_n.
\label{kp-cons}
\end{equation}
As a result we reach the follow theorem.

\begin{thm}\label{Thm-1}
Under symmetry constraint \eqref{u0-SC} where $Q_n$ and $R_n$ obey \eqref{QRx-SC}, the spectral problem \eqref{Lax-ta} is written as \eqref{sp-SC}.
\end{thm}

Note that in \cite{CheLH-JNMP-2013} \eqref{kp-cons} is called constrained discrete KP hierarchy,
which is from here actually a result of the symmetry constraint \eqref{u0-SC} together with \eqref{QRx-SC}.

\subsection{The sdAKNS hierarchy from symmetry constraint}\label{Sec-4-2}

There is a sdAKNS hierarchy coming from \eqref{Lax-tc} and its adjoint form
under the constraint \eqref{u0-SC}. This agrees with the continuum limit of the D$^2\Delta$KP hierarchy
and symmetry constraint of the continuous KP hierarchy (cf.\cite{KS-IP-1991,CL-JPA-1992,Chen-book-2006})

In the following we prove
\begin{thm}\label{Thm-1-1}
For the pseudo-difference operator
\begin{equation}
\mathfrak{L}=\Delta-Q_nR_n-Q_n\Delta^{-1}R_n,
\label{L-sc}
\end{equation}
define $A_m=(\mathfrak{L}^m)_{+}$ and $A^*_m$ to be the adjoint operator of $A_m$ w.r.t. the bilinear form \eqref{bil}.
Then
\begin{subequations}
\begin{align}
& Q_{n,t_m}=A_m Q_n,\label{Q-eq}\\
& R_{n,t_m}=-A^*_m R_n \label{R-eq}
\end{align}
\label{QR-eq}
\end{subequations}
provide a recursive relation
\begin{equation}
\left(\begin{array}{c}Q_n\\R_n\end{array}\right)_{t_{m+1}}
=\mathcal{L}^{(+)}\left(\begin{array}{c}Q_n\\R_n\end{array}\right)_{t_{m}},
\label{hie-sc}
\end{equation}
where
\begin{equation}
{\mathcal{L}}^{(+)}=\left(
       \begin{array}{cc}
         \Delta-Q_n(E+1)\Delta^{-1} R_n  - Q_n R_n & -Q_n(E+1)\Delta^{-1} Q_n \\
         R_n(E+1)\Delta^{-1} R_n & -\Delta E^{-1}+R_n(E+1)\Delta^{-1} Q_n  - Q_n R_n
       \end{array}
     \right).
\label{L-rec+}
\end{equation}
\eqref{hie-sc} generates a sdAKNS hierarchy (see \eqref{hie-iso-sdakns+}).
\end{thm}

To prove the theorem we first give two lemmas.

\begin{lem}\label{Lem-SC-1}
Suppose that $p^{(m)}_{-1,n}=\mathrm{Res}_{\Delta}\mathfrak{L}^m$, i.e. $p^{(m)}_{-1,n}$ is the coefficient of $\Delta^{-1}$ term in $\mathfrak{L}^m$.
Then we have
\begin{equation}
\Delta p^{(m)}_{-1,n}=-(Q_n R_n)_{t_m}.
\label{p-1}
\end{equation}
\end{lem}

\begin{proof}
Comparing the constant terms of the left and right hand sides of the Lax equation
\[\mathfrak{L}_{t_m}=[A_m,\mathfrak{L}]=-[(\mathfrak{L}^m)_{-}, \mathfrak{L}],\]
one immediately obtains \eqref{p-1}.
Here $(\mathfrak{L}^m)_{-}=\mathfrak{L}^m-A_m$.
\end{proof}

\begin{lem}\label{Lem-SC-2}
The following relations hold,
\begin{align}
& (Q_n\Delta^{-1} R_n A_m)_{+}= Q_n\Delta^{-1} R_n A_m -Q_n\Delta^{-1} (A^*_m R_n), \label{Am-1}\\
& (A_m Q_n\Delta^{-1} R_n)_{+}= A_m Q_n\Delta^{-1} R_n -(A_m Q_n)\Delta^{-1} R_n. \label{Am-2}
\end{align}
\end{lem}

\begin{proof}
We prove them one by one. For \eqref{Am-1} we only need to prove
\begin{equation}
(Q_n\Delta^{-1} R_n A_m)_{-} = Q_n\Delta^{-1} (A^*_m R_n).
\label{eq1}
\end{equation}
In fact, supposing that $A_m=\sum^{m}_{j=0}a_{j,n}\Delta^{m-j}$ and noting that $R_{n}\to 0$ as $|n|\to \infty$, we have
\begin{align*}
(Q_n\Delta^{-1} R_n A_m)_{-} = & \Bigl( Q_n \Delta^{-1} R_n \sum^{m}_{j=0}a_{j,n}\Delta^{m-j}\Bigr)_{-}\\
 = & \Bigl[ Q_n \sum^{m}_{j=0}\sum ^{\infty}_{s=1}(-1)^{s-1}(\Delta^{s-1}E^{-s} R_n a_{j,n})\Delta^{m-j-s}\Bigr]_{-}\\
% = &Q_n \sum^{m}_{j=0}\sum ^{\infty}_{s=m-j+1}(-1)^s(\Delta^{s-1}E^{-s} R_n a_{j,n})\Delta^{m-j-s}  \\
 = &Q_n \sum ^{\infty}_{l=1}\Bigl[ (-1)^{l-1} \Delta^{l-1}E^{-l} \sum^{m}_{j=0}(-1)^{m-j}(\Delta^{m-j}E^{-(m-j)}a_{j,n} R_n )\Bigr]\Delta^{-l}  \\
 = &Q_n \sum ^{\infty}_{l=1}\Bigl[ (-1)^{l-1} \Delta^{l-1}E^{-l} (A^*_m R_n )\Bigr]\Delta^{-l}  \\
 = &Q_n\Delta^{-1}\Delta  \sum ^{\infty}_{l=1}\Bigl[ (-1)^{l-1} \Delta^{l-1}E^{-l} (A^*_m R_n )\Bigr]\Delta^{-l}  \\
 = &Q_n\Delta^{-1} \sum ^{\infty}_{l=1}\Bigl\{ [(\Delta^*)^{l-1} (A^*_m R_n) ]\Delta^{-l+1}-[(\Delta^*)^{l}(A^*_m R_n) ]\Delta^{-l}\Bigr\}  \\
 = &Q_n\Delta^{-1}(A^*_m R_n),
\end{align*}
i.e. \eqref{Am-1}.

Next, we prove \eqref{Am-2}. To calculate $(A_m Q_n\Delta^{-1} R_n)_{-}$ we rewrite the operator $A_m Q_n$ as
as a form of pure difference operator $A_m Q_n=\sum^m_{j=0}b_{j,n}\Delta^{m-j}$
in which only the constant term $b_{m,n}$ makes sense in $(A_m Q_n\Delta^{-1} R_n)_{-}$.
Since $b_{m,n}=(A_m Q_n)$ we immediately find $(A_m Q_n\Delta^{-1} R_n)_{-}=(A_m Q_n)\Delta^{-1} R_n$,
which leads to the relation \eqref{Am-2}.
\end{proof}

The lemma also implies a relation\cite{YLZ-JPA-2009}
\[[A_m, Q_n\Delta^{-1} R_n]_{-}= (A_m Q_n)\Delta^{-1} R_n-Q_n\Delta^{-1} (A^*_m R_n).\]

Now we come to the proof of Theorem \ref{Thm-1-1}.

\vskip 5pt
\noindent
\textit{Proof of Theorem \ref{Thm-1-1}.}
First,
\begin{align*}
A_{m+1}=& \Bigl[(\Delta-Q_nR_n-Q_n\Delta^{-1}R_n)(A_m+p^{(m)}_{-1,n}\Delta^{-1})\Bigr]_{+}\\
       =& \Delta A_m - Q_nR_n A_m - [E \Delta^{-1}(Q_n R_{n,t_m}+R_n Q_{n,t_m})]- (Q_n\Delta^{-1}R_n A_m)_{+},
\end{align*}
where we have made use of Lemma \ref{Lem-SC-1}. Substituting \eqref{Am-1} into the above we find
\begin{equation}
A_{m+1}= \Delta A_m - Q_nR_n A_m - Q_n\Delta^{-1}R_n A_m - Q_n\Delta^{-1} R_{n,t_m}- [E\Delta^{-1}(Q_n R_{n,t_m}+R_n Q_{n,t_m})].
\label{A-m+1}
\end{equation}
Note that the last term is a scalar.

Next, we calculate $A_{m+1}$ in another way:
\begin{align*}
A_{m+1}=& \Bigl[(A_m+p^{(m)}_{-1,n}\Delta^{-1})(\Delta-Q_nR_n-Q_n\Delta^{-1}R_n)\Bigr]_{+}\\
       =& A_m \Delta- A_m Q_nR_n +p^{(m)}_{-1,n} -(A_m Q_n\Delta^{-1}R_n)_{+}\\
       =& A_m \Delta- A_m Q_nR_n -A_m Q_n\Delta^{-1} R_n +Q_{n,t_m}\Delta^{-1} R_n-[\Delta^{-1}(Q_n R_{n,t_m}+R_n Q_{n,t_m})],
\end{align*}
where we have made use of Lemma \ref{Lem-SC-1} and \eqref{Am-2}.
Its adjoint form reads
\begin{equation}
A_{m+1}^* =\Delta^* A^*_m - Q_nR_n A^*_m + R_n E\Delta^{-1}Q_n A^*_m  -R_n E \Delta^{-1}Q_{n,t_m}-[\Delta^{-1}(Q_n R_{n,t_m}+R_n Q_{n,t_m})].
\label{A-m+1*}
\end{equation}

Now, imposing \eqref{A-m+1} on $Q_n$ and \eqref{A-m+1*} on $R_n$, respectively,
and making use of \eqref{QR-eq},
we arrive at the recursive relation \eqref{hie-sc}. Thus we complete the proof.

~ \hfill $\square$

Here we  remark that it is possible to prove that $(\varphi_n\b \varphi_n)_x$ is a symmetry of the whole D$^2\Delta$KP hierarchy.
Following the idea in \cite{Oev-PA-1993} on  additional symmetry
\begin{equation}
\mathfrak{L}_z=-[\varphi_n\Delta^{-1}\b\varphi_n, \mathfrak{L}],
\label{sym-add}
\end{equation}
where $\mathfrak{L}$ is the pseudo-difference operator \eqref{L},
$\varphi_n$ and $\b\varphi_n$ respectively satisfy \eqref{Lax-tc} and its adjoint form $\b\varphi_{n,t_j}=-A^*_j \b\varphi_n$,
and we additionally  request $\varphi_n$ and $\b\varphi_n$ satisfy \eqref{phi-x} as well due to the Lax triad \eqref{Lax-t},
equation \eqref{sym-add} yields for $\Delta^0$ term that
$u_{0,n,z}=\varphi_{n+1}\b\varphi_{n}-\varphi_{n}\b\varphi_{n-1}$,
which is $u_{0,n,z}=(\varphi_{n}\b\varphi_{n})_x$ under \eqref{phi-x}.
It has been proved that \cite{LCTLH-MM-2016} $[\partial_{t_j},\partial_z]\mathfrak{L}=0$,
which means $(\varphi_{n}\b\varphi_{n})_x$ and the D$^2\Delta$KP flows $\mathcal{K}_j$ defined in \eqref{DDKP-hie}
commute, i.e. $\llbr \mathcal{K}_j, (\varphi_{n}\b\varphi_{n})_x \rrbr=0$.
Thus, $(\varphi_{n}\b\varphi_{n})_x$ is a symmetry of the whole D$^2\Delta$KP hierarchy \eqref{DDKP-hie}.
Such a result can also be proved as in \cite{MSS-JMP-1990} for continuous case from another approach.
With this symmetry, the extended  D$^2\Delta$KP hierarchy \cite{YLZ-JPA-2009} can be interpreted as
a kind of symmetry constraints and then the sources provided by $(\varphi_{n}\b\varphi_{n})_x$  are automatically self-consistent.
We also remark that, as we can see, the additional condition \eqref{phi-x}, coming from the item in Lax triad for the independent variable $x$,
plays a significant role.
Anyway, by means of symmetry constraint \eqref{u0-SC},
one may study integrability properties of the sdAKNS hierarchy from a viewpoint of the D$^2\Delta$KP hierarchy.
This will be considered elsewhere.

\section{The sdAKNS hierarchies related to \eqref{sp-DT}}\label{Sec-5}

In this section we will first list isospectral and nonisospectral flows together with their Lie algebra derived from \eqref{sp-DT} in \cite{ZTOF-JMP-1991}.
As new results, from Lie algebraic structures of these flows we construct new symmetries of nonisospectral equations.
In addition, by considering infinite dimensional subalgebras and continuum limit of recursion operator,
we construct three types of isospectral and nonisospectral sdAKNS hierarchies.

\subsection{Flows related to \eqref{sp-DT} and their Lie algebra}\label{Sec-5-1}

In Ref.\cite{ZTOF-JMP-1991}, from spectral problem \eqref{sp-DT},
the following isospectral hierarchy $U_{n,t_s}=K_s$ and nonisospectral hierarchy $U_{n,t_s}=\sigma_s$ were derived:
\begin{equation}\label{hie-iso}
     U_{n,t_s} = K_s=L^s K_0,~~ K_0=\left(
                       \begin{array}{c}
                         Q_n \\
                         -R_n \\
                       \end{array}
                     \right), \quad s \in \mathbb{Z},
\end{equation}
\begin{equation}\label{hie-non}
     U_{n,t_s} = \sigma_s=L^s \sigma_0,~~ \sigma_0=\left(
                       \begin{array}{c}
                         (n+\frac{1}{2})Q_n \\
                         -(n-\frac{1}{2})R_n \\
                       \end{array}
                     \right), \quad s \in \mathbb{Z},
\end{equation}
where $U_n=(Q_n,R_n)^T$, $L$ is a recursion operator
\begin{equation}\label{rec-L}
  L= \left(
       \begin{array}{cc}
         E & 0 \\
         0 & E^{-1} \\
       \end{array}
     \right) -\left(
                         \begin{array}{c}
                           Q_n \\
                           -R_n \\
                         \end{array}
                       \right) (E+1)\Delta^{-1} (R_n,Q_n) - Q_n R_n,
\end{equation}
and its inverse $L^{-1}$ is
\begin{equation}\label{rec-L-1}
 L^{-1}= \left(
           \begin{array}{cc}
             E^{-1}\mu_n^{-1} & 0 \\
             0 & \mu_n^{-1} E \\
           \end{array}
         \right) + \left(
                     \begin{array}{c}
                       E^{-1}Q_n\mu_n^{-1}\\
                       -R_{n+1}\mu_n^{-1} \\
                     \end{array}
                   \right)(E+1)\Delta^{-1}( \mu_n^{-1} R_{n+1},\mu_n^{-1} Q_n E )
\end{equation}
with $\mu_n = 1+Q_nR_{n+1}$.
$L$ is a hereditary operator. The simplest ispospectral flows and nonisospectral flows are
\begin{subequations}
\begin{align}
& K_{-1} =  \left(
                               \begin{array}{c}
                                 {Q_{n-1}}/{\mu_{n-1}} \\
                                 -{R_{n+1}}/{\mu_n} \\
                               \end{array}
                             \right),~~~
K_{0} = \left(
                \begin{array}{c}
                  Q_n \\
                  -R_n \\
                \end{array}
              \right),
~~~K_{1} =  \left(
                           \begin{array}{c}
                             Q_{n+1}-Q_n^2R_n \\
                             -R_{n-1}+R_n^2Q_n \\
                           \end{array}
                         \right),\\
& K_{2} =  \left(
                        \begin{array}{c}
                             Q_{n+2}-Q_{n+1}^2 R_{n+1}-Q_n^2 R_{n-1}- 2Q_n Q_{n+1} R_n+ Q_n^3 R_n^2 \\
                             - R_{n-2}+R_{n-1}^2 Q_{n-1}+R_n^2 Q_{n+1}+ 2R_{n-1} R_n Q_n- R_n^3 Q_n^2 \\
                         \end{array}
                      \right),
\end{align}
\end{subequations}
and
\begin{subequations} \label{non-flows}
 \begin{align}
& \sigma_{-1} =  \left(
                                     \begin{array}{c}
                                       (n- 1/2)Q_{n-1}/\mu_{n-1} \\
                                     -(n+ 1/2)R_{n+1}/\mu_{n}  \\
                                     \end{array}\right),~~~
\sigma_{0} = \left(
                       \begin{array}{c}
                         (n+ 1/2)Q_n \\
                         -(n- 1/2)R_n \\
                       \end{array}
                     \right),             \label{non-flow-a}  \\
& \sigma_{1} = \left(
                                \begin{array}{c}
                                  (n+3/2)Q_{n+1} - (n+3/2)Q_n^2 R_n - 2 Q_n \Delta^{-1}Q_nR_n \\
                                 - (n-3/2)R_{n-1} + (n+1/2)R_n^2 Q_n + 2 R_n \Delta^{-1}R_nQ_n \\
                                \end{array}
                              \right).
                  \label{non-flow-b}
 \end{align}
\end{subequations}

The flows $K_s$ and $\sigma_k$ defined in \eqref{hie-iso} and \eqref{hie-non} constitute a
centerless Virasoro algebra \cite{ZTOF-JMP-1991},
\begin{subequations}\label{alg-k-s}
\begin{align}
&\llbr K_m,K_s \rrbr =0,\\
&\llbr K_m,\sigma_s \rrbr = m K_{m+s},\\
&\llbr \sigma_m,\sigma_s \rrbr =(m-s)\sigma_{m+s},~~~ m,s\in \mathbb{Z}.
\end{align}
\end{subequations}

\subsection{New symmetries of nonisospectral equations \eqref{hie-non}}\label{Sec-5-2}

Note that algebraic structure \eqref{alg-k-s} is the same as the one generated by the AL flows (see \cite{ZNBC-PLA-2006}).
This means the hierarchies \eqref{hie-iso} and \eqref{hie-non} and the AL hierarchies can share those results
obtained from the  algebraic structure \eqref{alg-k-s}.
One remarkable result is that the nonisospectral hierarchy \eqref{hie-non} posses symmetries.
In fact, it is very rare for a nonisospectral equation to have infinitely many symmetries. However,
making use of minus indices in the structure \eqref{alg-k-s}, infinitely many symmetries for the  nonisospectral hierarchy \eqref{hie-non}
can be constructed.

\begin{thm}
Any given member $U_{n,t_m}=\sigma_m$ in the nonisospectral hierarchy \eqref{hie-non} possesses two sets of symmetries,
\begin{subequations}
\begin{align}
&  \eta^{(m)}_{s}=\sum_{j=0}^s \mathrm{C}_s^j(mt_m)^{s-j}\sigma_{m-jm},~~~(s=0,1,2\cdots.),\\
& \gamma^{(m)}_{s}=\sum_{j=0}^s \mathrm{C}_s^j(mt_m)^{s-j}K_{-jm},~~~(s=0,1,2\cdots.),
\end{align}
\end{subequations}
and these symmetries form a centerless Virasoro algebra with  structure
\begin{subequations}
\begin{align}
&\llbr \gamma^{(m)}_{l},\gamma^{(m)}_{s} \rrbr =0,\\
&\llbr \gamma^{(m)}_{l},\eta^{(m)}_{s} \rrbr =-ml\gamma^{(m)}_{l+s-1},\\
&\llbr \eta^{(m)}_{l},\eta^{(m)}_{s}  \rrbr =-m(l-s)\eta^{(m)}_{l+s-1}.
\end{align}
\end{subequations}
Here the suffix $(m)$ corresponds to equation $U_{n,t_m}=\sigma_m$.
\end{thm}

We skip the proof, for which  one can refer to  Proposition 5.1 in \cite{ZNBC-PLA-2006}.

For an isospectral equation $U_{n,t_m}=K_m$ in the isospectral hierarchy \eqref{hie-iso},
it also has two sets of symmetries,
\begin{eqnarray*}
\{K_s\} \text{ and~~} \{\tau^{(m)}_s=mt_mK_{m+s}+\sigma_s\},~~~s\in \mathbb{Z}.
\end{eqnarray*}
and they  form a centerless Virosoro algebra as well, with structure
%\begin{subequations}
\begin{eqnarray*}
&& \llbr K_l,K_s\rrbr=0,\\
&& \llbr K_l,\tau^{(m)}_s\rrbr =l K_{l+s},\\
&& \llbr \tau^{(m)}_l,\tau^{(m)}_s \rrbr =(l-s)\tau^{(m)}_{l+s}.
\end{eqnarray*}
%\end{subequations}

\subsection{Infinite dimensional subalgebras and new sdAKNS hierarchies}\label{Sec-5-3}

Equations \eqref{hie-iso} and \eqref{hie-non} are not the sdAKNS hierarchies.
It is interesting to consider infinite dimensional subalgebras of the algebra \eqref{alg-k-s}.
These subalgebras, together with continuum limits of the recursion operator \eqref{rec-L},
can be used to construct and identify  sdAKNS hierarchies.

\subsubsection{Infinite dimensional subalgebras}\label{Sec-5-3-1}

Define
\begin{align}
& \b K^{(+)}_s=({\mathcal{L}}^{(+)})^s K_0,~~ \b \sigma^{(+)}_s=({\mathcal{L}}^{(+)})^s \sigma_0,~~ s=0,1,2,\cdots, \label{flow-K+}
\\
& \b K^{(-)}_s=({\mathcal{L}}^{(-)})^s K_0,~~ \b \sigma^{(+)}_s=({\mathcal{L}}^{(-)})^s \sigma_0,~~ s=0,1,2,\cdots, \label{flow-K-}
\end{align}
and
\begin{subequations}
\begin{align}
& \b K_{2s+1}={\mathcal{L}}^s (K_{1}-K_{-1})/2,~~ \b K_{2s}=\mathcal{L}^{s} K_0,~~ s=0,1,2,\cdots,\\
& \b \sigma_{2s+1}= {\mathcal{L}}^s (\sigma_{1}-\sigma_{-1})/2,~~ \b \sigma_{2s}=\mathcal{L}^{s} \sigma_0,~~ s=0,1,2,\cdots,
\end{align}
\label{flow-KK}
\end{subequations}
where
\begin{equation}\label{LL}
{\mathcal{L}}=L-2I+L^{-1},~~ {\mathcal{L}}^{(+)}=L-I,~~{\mathcal{L}}^{(-)}=I-L^{-1}
\end{equation}
and $I$ is the $2 \times 2$ unit matrix.

\begin{lem}\label{prop-4-1}
The flows
\begin{equation}
\mathrm{(I):}~ \{ \b K^{(+)}_s, \b \sigma^{(+)}_l\},~~~
\mathrm{(II):}~ \{ \b K^{(-)}_s, \b \sigma^{(-)}_l\},~~~
\mathrm{(III):}~ \{ \b K_{2m+j}, \b \sigma_{2s+k}\},
\end{equation}
generate three infinite dimensional subalgebras, respectively, with structures
\begin{align*}
\mathrm{(I):}~& \llbr  \b K^{(+)}_s, \b K^{(+)}_l \rrbr =0, \\
& \llbr  \b K^{(+)}_s, \b \sigma^{(+)}_l \rrbr =s(\b K^{(+)}_{s+l} +\b K^{(+)}_{s+l-1}), \\
& \llbr  \b \sigma^{(+)}_s, \b \sigma^{(+)}_l \rrbr =(s-l)( \b \sigma^{(+)}_{s+l} +  \b \sigma^{(+)}_{s+l-1}),
\end{align*}
\begin{align*}
\mathrm{(II):}~& \llbr  \b K^{(-)}_s, \b K^{(-)}_l \rrbr =0, \\
& \llbr  \b K^{(-)}_s, \b \sigma^{(-)}_l \rrbr =-s(\b K^{(-)}_{s+l} - \b K^{(-)}_{s+l-1}), \\
& \llbr  \b \sigma^{(-)}_s, \b \sigma^{(-)}_l \rrbr =-(s-l)(\b \sigma^{(-)}_{s+l}- \b \sigma^{(-)}_{s+l-1}),
\end{align*}
\begin{align*}
\mathrm{(III):}~& \llbr \b K_{2m+j}, \b K_{2s+k}\rrbr =0,\\
& \llbr \b K_{2m}, \b \sigma_{2s}\rrbr = 2m \b K_{2(m+s)-1},\\
& \llbr \b K_{2m}, \b \sigma_{2s+1}\rrbr = 2m \b K_{2(m+s)}+ \frac{m}{2} \b K_{2(m+s+1)},\\
& \llbr \b K_{2m+1}, \b \sigma_{2s+k}\rrbr = (2m+1) \b K_{2(m+s)+k}+ \frac{m+1}{2} \b K_{2(m+s+1)+k},\\
& \llbr \b \sigma_{2m}, \b \sigma_{2s}\rrbr = 2(m-s) \b \sigma_{2(m+s)-1},\\
& \llbr \b \sigma_{2m+j}, \b \sigma_{2s+1}\rrbr =[2(m-s)-1+j] \b \sigma_{2(m+s)+j}+ \frac{m-s-1+j}{2} \b \sigma_{2(m+s+1)+j},
\end{align*}
where $j,k\in\{0,1\}$, $m,s\geq 0$ and we define $\b K^{(\pm)}_{-1} =\b \sigma^{(\pm)}_{-1}  =\b K_{-1}=\b \sigma_{-1}=0$.
Obviously, the set (III) has a subalgebra $\{ \b K_{2m+1}, \b \sigma_{2s+1}\}$.
\end{lem}

This lemma can be verified directly. The structure of set (III) has been proved in \cite{ZhaC-SAM-2010-II}.
%Besides,  set (III) has a subalgebra composed by $\{\b K_{2m+1}, \b \sigma_{2s+1}\}$.
Note that operator $L-L^{-1}$ does not generate a subalgebra of \eqref{alg-k-s}.

\subsubsection{Three sdAKNS hierarchies}\label{Sec-5-3-2}

Let us construct the sdAKNS hierarchies through considering possible combinations of the flows $\{K_s\}$ and $\{\sigma_s\}$
defined in \eqref{hie-iso} and \eqref{hie-non}.
The criteria is that these combined flows should be closed as a subalgebra of \eqref{alg-k-s}
and they must yield their counterparts in  the continuous AKNS flows in reasonable continuum limits.
For this purpose we investigate continuum limits of initial flows, $L$ and $L^{-1}$ under a uniform scheme\footnote{The correspondence
between $x$ and $n$ is $x=x_0+n\epsilon$ where $\epsilon$ is viewed as a spacing parameter.
Here we take $x_0=0$ for convenience.}
\begin{equation}
 U_{n+j}=\epsilon(q(x+j\epsilon,t), r(x+j\epsilon,t))^T,~~~
 n\epsilon = x, (n\to \infty, \epsilon\to 0).
\label{cl-sche}
\end{equation}
We find
\begin{subequations}
\begin{align}
& K_0=\epsilon (q,-r)^T,~~ (K_{1}-K_{-1})/2=\epsilon^2 (q,r)_x^T + O(\epsilon^3),\\
& \sigma_0= x (q,-r)^T + O(\epsilon),~~ (\sigma_{1}-\sigma_{-1})/2=\epsilon (xq_x+q,xr_x+r)^T + O(\epsilon^2),
\end{align}
\label{cl-K0s0}
\end{subequations}
and
\begin{subequations}
\begin{align}
& L=I + \epsilon L_{AKNS} +\frac{\epsilon^2}{2}L_{NLS}+ O(\epsilon^3),\\
& L^{-1} =I - \epsilon L_{AKNS} +\frac{\epsilon^2}{2}(2L^2_{AKNS}-L_{NLS})+ O(\epsilon^3),
\end{align}
\label{cl-LL-1}
\end{subequations}
where
\begin{equation}
L_{AKNS}=\left(\begin{array}{cc}
            \partial_x- 2 q\partial_x^{-1}r & -2q \partial^{-1}_x q\\
            2r \partial^{-1}_x r & -\partial_x+2 r\partial_x^{-1}q
            \end{array}\right)
\label{rec-L-akns}
\end{equation}
is the recursion operator of the continuous AKNS hierarchy,
and
\begin{equation*}
L_{NLS}=\left(\begin{array}{cc}
            \partial_x^2  & 0\\
            0 & \partial_x^2
            \end{array}\right)-2qr,
\end{equation*}
which yields the nonlinear Schr\"odinger system by acting on $(q,-r)^T$.
\eqref{cl-LL-1} indicates
\begin{subequations}
\begin{align}
& {\mathcal{L}}^{(+)}=L-I= \epsilon L_{AKNS} +O(\epsilon^2),\\
& {\mathcal{L}}^{(-)}=I-L^{-1}= \epsilon L_{AKNS} +O(\epsilon^2),\\
& {\mathcal{L}}=L-2I+L^{-1}= \epsilon^2 L^2_{AKNS} +O(\epsilon^3).
\end{align}
\end{subequations}
Thus, based on the continuum limits of initial flows in \eqref{cl-K0s0},
the definition of the flows (\ref{flow-K+},\ref{flow-K-},\ref{flow-KK}) and Lemma \ref{prop-4-1},
we obtain three sets of sdAKNS hierarchies.

\begin{thm}\label{prop-4-2}
The flows defined in (\ref{flow-K+},\ref{flow-K-},\ref{flow-KK}) yield three sets of sdAKNS hierarchies
\begin{align}
\mathrm{(I):}~& U_{t_s}= \b K^{(+)}_s,~~~ U_{t_s}=  \b \sigma^{(+)}_s,\label{hie-iso-sdakns+}\\
\mathrm{(II):}~& U_{t_s}= \b K^{(-)}_s,~~~ U_{t_s}= \b \sigma^{(-)}_s,\label{hie-iso-sdakns-}\\
\mathrm{(III):}~& U_{t_s}= \b K_{s}, ~~~~~ U_{t_s}=  \b \sigma_{s}, \label{hie-iso-sdakns}
\end{align}
where $s=0,1,\cdots$. They all correspond to the continuous isospectral and nonisospectral AKNS hierarchies
under the continuum limit \eqref{cl-sche}\footnote{In principle we
need to suitably rescale $t_s$ by $\epsilon^j t_s$. For example, for $U_{t_s}= \b K_{s}$, rescale $t_s$ by $\epsilon^{-s} t_s$.}
\end{thm}

Here we remark that $\{U_{t_s}= \b K^{(+)}_s\}$ was already found \cite{MRT-IP-1994},
which is just \eqref{hie-sc}, the result of symmetry constraint of the D$^2\Delta$KP hierarchy. Besides, equation  $U_{t_2}= \b K_{2}$
was also mentioned in \cite{MRT-IP-1994} as a discretization of the 2nd order AKNS equations.

\section{Conclusions}\label{Sec-6}

The spectral problem \eqref{sp-DT} can generate a sdAKNS hierarchy. It is also a Darboux transformation of the ZS-AKNS spectral problem \eqref{sp-AKNS}.
By revisiting it, we have shown that it is gauge equivalent to \eqref{sp-1} which provides a bidirectional discretization of \eqref{sp-AKNS},
while the AL spectral problem \eqref{sp-AL}  comes from a monodirectional discretization of \eqref{sp-AKNS}.
As a relation with higher dimensional systems, we proved that  \eqref{sp-DT} and a related sdAKNS hierarchy can be obtained from the Lax triads
of the D$^2\Delta$KP hierarchy via the symmetry constraint \eqref{u0-SC}.
This fact, on one side, coincides with the continuous case \cite{KS-IP-1991,CL-JPA-1992}.
On the other side, it exhibits a new aspect of discrete systems:
there are two discrete spectral problems, \eqref{sp-AL} and \eqref{sp-DT}
which can generate sdAKNS hierarchy, but only \eqref{sp-DT} that is a bidirectional discretisation of the ZS-AKNS spectral problem
is related to the symmetry constraint of the D$^2\Delta$KP hierarchy.
In addition to the above results, three sdAKNS hierarchies \eqref{hie-iso-sdakns+}, \eqref{hie-iso-sdakns-} and \eqref{hie-iso-sdakns}
are obtained with a criteria  that the corresponding flows are closed w.r.t. Lie product \eqref{def:comm}
and in continuum limit they approach to the continuous AKNS hierarchy.
Among these sdAKNS hierarchies, $\{U_{n,t_j}=\b K^{(+)}_j\}$ is the one derived from Lax triad of the D$^2\Delta$KP hierarchy via the symmetry constraint.

With regard to the symmetry-constrainted spectral problem \eqref{sp-SC}, in  \cite{LLTHC-JMP-2013}
a spectral problem
\begin{equation}
\h{\mathfrak{L}}\phi_n=\xi \phi_n,~~
\h{\mathfrak{L}}=\Delta + Q_n\Delta^{-1}R_n
\label{sp-SC-LLTHC}
\end{equation}
was given as a constrain of $\h{\mathfrak{L}}\phi_n=\xi \phi_n,~
\h{\mathfrak{L}}=\Delta + u_{1,n}\Delta^{-1}+ u_{2,n}\Delta^{-1}+u_{2,n}\Delta^{-1}+\cdots$,
and the related equations were investigated (e.g. \cite{LCH-MPLB-2013}).
However, our results are different from them.

There are several interesting problems that could be followed. Apart from
continuous correspondence of integrability properties of the new sdAKNS hierarchies£¬
there would be many interactions between the (1+1) and (2+1)-dimensional systems based on symmetry constraints,
such as solutions and integrability characteristics (e.g.\cite{CL-PLA-1991,KS-IP-1991,Oev-PA-1993,SS-IP-1991} in continuous case).
Some known results related to the pseudo-difference operator \eqref{L} (e.g. \cite{HI-IMRN-2000,LC-JPA-2010}) could also be used to
investigate the sdAKNS hierarchies.

\vskip 15pt
\subsection*{Acknowledgments}
This project is  supported by the NSF of China (Nos.11371241 and 11631007).

\vskip 15pt

\begin{appendix}

\section{The sdAKNS hierarchies from the AL spectral problem }\label{Sec-app}

From the AL spectral problem \eqref{sp-AL} one can derive the AL hierarchy (cf.\cite{FQSZ-arxiv-2013}):
\begin{align}\label{AL-hie}
U_{t_s}=K_s=\b L^s  K_0,~~ s\in \mathbb{Z},
\end{align}
where $K_0$ and the first few flows are
%\bse\label{AL-Ks}
\begin{align*}
& K_0=
\left(
  \begin{array}{c}
    Q_n \\
    -R_n \\
  \end{array}
\right),~~
 K_1=\b \mu_n
\left(
  \begin{array}{c}
    Q_{n+1} \\
    -R_{n-1} \\
  \end{array}
\right),~~
 K_{-1}=\b \mu_n
\left(
  \begin{array}{c}
    Q_{n-1} \\
    -R_{n+1} \\
  \end{array}
\right),
%& K_2=\mu_n
%\left(
%  \begin{array}{c}
%    \mu_{n+1}Q_{n+2}-Q_{n+1}(Q_{n}R_{n-1}+Q_{n+1}R_{n}) \\
%    -\mu_{n-1}Q_{n-2}+R_{n-1}(Q_{n+1}R_{n}+Q_{n}R_{n-1}) \\
%  \end{array}
%\right),\label{AL-K(+2)}\\
%& K_{-2}=\mu_n
%\left(
%  \begin{array}{c}
%    \mu_{n-1}Q_{n-2}-Q_{n-1}(Q_{n}R_{n+1}+Q_{n-1}R_{n}) \\
%    -\mu_{n+1}Q_{n+2}+R_{n+1}(Q_{n-1}R_{n}+Q_{n}R_{n+1}) \\
%  \end{array}
%\right),\label{AL-K(-2)}
\end{align*}
%\ese
the recursion operator reads
\begin{align}\label{AL-L}
\b L=&
\left(
  \begin{array}{cc}
    E & 0 \\
    0 & E^{-1} \\
  \end{array}
\right)+
\left(
  \begin{array}{c}
    -Q_nE \\
    R_n \\
  \end{array}
\right)
\Delta^{-1}(R_nE,Q_nE^{-1})\nonumber\\
&+\b\mu_n
\left(
  \begin{array}{c}
    -EQ_n \\
    R_{n-1} \\
  \end{array}
\right)
\Delta^{-1}(R_n,Q_n)\frac{1}{\b\mu_n},
\end{align}
with its inverse
\begin{align*}
\b L^{-1}=&
\left(
  \begin{array}{cc}
    E^{-1} & 0 \\
    0 & E \\
  \end{array}
\right)+
\left(
  \begin{array}{c}
    Q_n \\
    -R_nE \\
  \end{array}
\right)
\Delta^{-1}(R_nE^{-1},Q_nE)\nonumber\\
&+\b\mu_n
\left(
  \begin{array}{c}
    Q_{n-1} \\
    -ER_n \\
  \end{array}
\right)
\Delta^{-1}(R_n,Q_n)\frac{1}{\b\mu_n},
\end{align*}
and here $\b\mu_n=1-Q_nR_n$.
Under the continuum limit scheme \eqref{cl-sche}, one can find
\begin{equation*}
 \b L=I + \epsilon L_{AKNS} +\frac{\epsilon^2}{2}L^2_{AKNS}+ O(\epsilon^3),
\end{equation*}
where $L_{AKNS}$ is given in \eqref{rec-L-akns}.

There are also three sdAKNS hierarchies related to the AL spectral problem:
\begin{equation}
\mathrm{I}: \{U_{n,t_s}=\b K^{(+)}_s\},~~ \mathrm{II}: \{ U_{n,t_s}=\b K^{(-)}_s\},~~ \mathrm{III}: \{U_{n,t_s}=\b K_s\},
\label{hie-iso-akns-al}
\end{equation}
where
\begin{align}
& \b K^{(+)}_s=(\b{\mathcal{L}}^{(+)})^s K_0,\\
& \b K^{(-)}_s=(\b{\mathcal{L}}^{(-)})^s K_0,\\
& \b K_{2s+1}=\b{\mathcal{L}}^s (K_{1}-K_{-1})/2,~~ \b K_{2s}=\mathcal{L}^{s} K_0,~~ s=0,1,2,\cdots,
\end{align}
\begin{equation}\label{b-LL}
\b{\mathcal{L}}=\b L-2I+\b L^{-1},~~ \b{\mathcal{L}}^{(+)}=\b L-I,~~\b{\mathcal{L}}^{(-)}=I-\b L^{-1}.
\end{equation}
The third sdAKNS hierarchy has been well studied and it admits one-field reduction to get sdKdV, sdmKdV and sdNLS hierarchies.
As a review one can refer to \cite{FQSZ-arxiv-2013}.

\end{appendix}

%%%%%%%%%%%%%%%%%%%%%%%%%%%%%%%%%%%%%%%%%%%%%%%%%%%%%%%%%%%%%%%%%%%%%%%%%%%%%%%%%%%%%%%%%%%%%%%%%%%%%%%%%%%%%%%%%%%%%%%%%%%%%%


\begin{thebibliography}{00}

%\small

\bibitem{ZakS-JETP-1972} V.E. Zakharnov, A.B. Shabat,
         Exact theory of two-dimensional self-focusing and one-dimensional self-modulating waves in nonlinear media,
         Sov. Phys. JETP, 34 (1972) 62-9.

\bibitem{AKNS-PRL-1973}M.J. Ablowitz, D.J. Kaup, A.C. Newell, H. Segur,
         Nonlinear-evolution equations of physical significance,
         Phys. Rev. Lett., 31 (1973) 125-7.

\bibitem{AL-JMP-1975} M.J. Ablowitz, J.F. Ladik,
         Nonlinear differential-difference equations,
         J. Math. Phys., 16 (1975) 598-603.

\bibitem{AL-JMP-1976} M.J. Ablowitz, J.F. Ladik,
         Nonlinear differential-difference equations and Fourier analysis,
         J. Math. Phys., 17 (1976) 1011-8.

\bibitem{ZhaC-SAM-2010-I} D.J. Zhang, S.T. Chen,
         Symmetries for the Ablowitz-Ladik hierarchy: Part I. Four-potential case,
         Stud. Appl. Math., 125 (2010) 393-418.

\bibitem{ZhaC-SAM-2010-II} D.J. Zhang, S.T. Chen,
         Symmetries for the Ablowitz-Ladik hierarchy: Part II. Integrable discrete nonlinear Schr\"odinger equations and discrete AKNS hierarchy,
         Stud. Appl. Math., 125 (2010) 419-43.


\bibitem{FQSZ-arxiv-2013} W. Fu, Z.J. Qiao, J.W. Sun, D.J. Zhang,
         The semi-discrete AKNS system: Conservation laws, reductions and continuum limits,
         arXiv: 1307.3671.

\bibitem{RT-prin-1989} O. Ragnisco, G. Z. Tu,
         A new hierarchy of integrable discrete systems,
         preprint, Dept. of Phys., University of Rome, 1989.
         %(first consider the sp and get hierarchy)

\bibitem{ZTOF-JMP-1991} H.W. Zhang, G.Z. Tu, W. Oevel, B. Fuchssteiner,
         Symmetries, conserved quantities, and hierarchies for some lattice systems with soliton structure,
         J. Math. Phys., 32 (1991) 1908-18.
         %(consider the sp in [RT-prin-1989] and get hierarchy, recursion operator, Hamiltonian, master symmetry etc.
         %Obtained recursion operator and its formal inverse (4.9), (4.10))

\bibitem{MRT-IP-1994} I. Merola, O. Ragnisco, G.Z. Tu,
         A novel hierarchy of integrable lattices,
         Inverse Problems, 10 (1994) 1315-34.
         %(show the hierarchy in ZTOF-JMP-1991 is sdAKNS)


\bibitem{KSS-PLA-1991}B. Konopelchenko, J. Sidorenko, W. Strampp,
        (1+1)-dimensional integrable systems as symmetry constraints of (2+1)-dimensional systems,
        Phys. Lett. A, 157 (1991) 17-21.

\bibitem{KS-IP-1991} B. Konopelchenko, W. Strampp,
         The AKNS hierarchy as symmetry constraint of the KP hierarchy,
         Inverse Problems, 7 (1991) L17-24.

\bibitem{CL-PLA-1991} Y. Cheng, Y.S. Li,
         The constraint of the Kadomtsev-Petviashvili equation and its special solutions,
         Phys. Lett. A, 157 (1991) 22-6.

\bibitem{CL-JPA-1992} Y. Cheng, Y.S. Li,
         Constraints of the 2+1 dimensional integrable soliton systems,
         J. Phys. A: Math. Gen. 25 (1992) 419-31.

\bibitem{DJM-JPSJ-1982} E. Date, M. Jimbo, T. Miwa,
         Method of generating discrete soliton equations II,
         J. Phys. Soc. Jpn., 51 (1982) 4125-31.

\bibitem{KT-SCF-1997} S. Kanaga Vel, K.M. Tamizhmani,
         Lax pairs, symmetries and conservation laws of a differential-difference equation --- Sato's approach,
         Chaos, Solitons and Fractals, 8 (1997) 917-31.

\bibitem{LB-PNAS-1980} D. Levi, R. Benguria,
         B\"acklund transformations and nonlinear differential difference equations,
         Proc. Natl. Acad. Sci. U.S.A., 77 (1980) 5025-7.

\bibitem{Levi-JPA-1981} D. Levi,
         Nonlinear differential difference equations as B\"acklund transformations,
         J. Phys. A: Math. Gen., 14 (1981) 1083-98.

\bibitem{CZ-CPL-2012} C.W. Cao, G.Y. Zhang,
         Lax pairs for discrete integrable equations via Darboux transformations,
         Chin. Phys. Lett., 29 (2012) No.50202 (2pp).

\bibitem{CaoZ-JPA-2012} C.W. Cao, G.Y. Zhang,
         Integrable symplectic maps associated with the ZS-AKNS spectral problem,
         J. Phys. A: Math. Theor., 45 (2012) No.265201 (15pp).

\bibitem{KMW-TMP-2013} F. Khanizadeh, A.V. Mikhailov, J.P. Wang,
         Darboux transformations and recursion operators for differential-difference equations,
         Theore. Math. Phys., 177 (2013) 387-440.
         %(provide inverse of $L$ obtained in ZTOF-JMP-1991 and MRT-IP-1994)

\bibitem{Mik-INI-2013} A.V. Mikhailov,
         Formal diagonalisation of the Lax-Darboux scheme and conservation laws of integrable partial differential,
         differential-difference and partial difference equations,
         http://www.newton.ac.uk/files/seminar/20130711140014301-153657.pdf .

\bibitem{AY-JPA-1994} V.E. Adler, R.I. Yamilov,
         Explicit auto-transformations of integrable chains,
         J. Phys. A: Math. Gen., 27 (1994) 477-92.
         %(DT of NLS, which is sp in ZTOF-JMP-1991 and MRT-IP-1994)

\bibitem{FF-PD-1981} B. Fuchssteiner, A.S. Fokas,
         Symplectic structures, their B\"acklund transformations and hereditary symmetries,
         Physica D, 4 (1981) 47-66.

\bibitem{Fuc-MSE-1991} B. Fuchssteiner,
         Hamiltonian structure and integrability,
         Math. Sci. Eng., 185 (1991) 211-56.


%======== gauge ==============

\bibitem{FanY-IJTP-2009} E.G. Fan, Z.H. Yang,
         A lattice hierarchy with a free function and its reductions to the Ablowitz-Ladik and Volterra hierarchies,
         Int. J. Theor. Phys., 48 (2009) 1-9.

\bibitem{TZZZ-MPLB-2016} S.F. Tian, F.B. Zhou, S.W Zhou, T.T. Zhang,
         Analytic solutions and Darboux transformation to a new Hamiltonian lattice hierarchy,
         Mod. Phys. Lett. B, 30 (2016) No.1650100 (12pp).

\bibitem{DonZZ-CNSNS-2016} H.H Dong, Y.Z. Zhang, X.E. Zhang,
        The new integrable symplectic map and the symmetry of integrable nonlinear lattice equation,
        Commun. Nonl. Sci. Numer. Simulat., 36 (2016) 354-65.

\bibitem{YXD-CSF-2005} H.X. Yang, X.X. Xu, H.Y. Ding,
         New hierarchies of integrable positive and negative lattice models and Darboux transformation,
         Chaos, Solitons \& Fractals, 26 (2005) 1091-103.

\bibitem{DXYD-MPLB-2007} H.Y. Ding, X.X. Xu, H.X. Yang, X. Tian,
         A new hierarchy of discrete integrable model and its Darboux transformation,
         Mod. Phys. Lett. B, 21 (2007) 189-97.

\bibitem{Qin-JMP-2008} Z.Y. Qin,
         A generalized Ablowitz-Ladik hierarchy, multi-Hamiltonian structure and Darboux transformation,
         J. Math. Phys., 49 (2008) No.063505 (14pp).

\bibitem{WRH-DDNS-2012} X.B. Wu, W.G. Rui, X.C. Hong,
         A new discrete integrable system derived from a generalized Ablowitz-Ladik hierarchy and its Darboux transformation,
         Discrete Dyn. Nat. Soc., 2012 (2012) No.652076 (19pp).

\bibitem{LLC-CTP-2014} H.M. Li, Y.Q. Li, Y. Chen,
         An integrable discrete generalized nonlinear Schr\"odinger equation and its reductions,
         Commun. Theor. Phys., 62 (2014) 641-8.

%======== gauge =======

\bibitem{FHTZ-Non-2013} W. Fu, L. Huang, K.M. Tamizhmani, D.J. Zhang,
         Integrable properties of the differential-difference Kadomtsev-Petviashvili hierarchy and continuum limits,
         Nonlinearity, 26 (2013) 3197-229.

\bibitem{CaoZX-prn-2016} C.W. Cao, D.J. Zhang, X.X. Xu,
         On the lattice potential KP equation, preprint, 2016.

\bibitem{CheLH-JNMP-2013} J.P. Cheng, M.H. Li, J.S. He,
         The Virasoro action on the tau function for the constrained discrete KP hierarchy,
         J. Nonl. Math. Phys., 20 (2013) 529-38.


\bibitem{Chen-book-2006} D.Y. Chen,
         Introduction to Soliton Theory (in Chinese),
         Sci. Pub. House, Beijing, 2006.

\bibitem{YLZ-JPA-2009} Y.Q. Yao, X.J. Liu, Y.B. Zeng,
        A new extended discrete KP hierarchy and a generalized dressing method,
        J. Phys. A: Math. Theor., 42 (2009) No.454026 (10pp).

\bibitem{Oev-PA-1993} W. Oevel,
        Darboux theorems and  Wronskian formulas for integrable systems I: Constrainted KP flows,
        Physica A, 195 (1993) 533-76.

\bibitem{LCTLH-MM-2016} C.Z. Li, J.P. Cheng, K.L. Tian, M.H. Li, J.S. He,
        Ghost symmetry of the discrete KP hierarchy,
        Monatsh. Math., 180 (2016) 815-32. (arXiv: 1201.4419).



\bibitem{MSS-JMP-1990} J. Matsukidaira, J. Satsuma, W. Strampp,
         Conserved quantities and symmetries of KP hierarchy,
         J. Math. Phys., 31 (1990) 1426-34.



%\bibitem{Mag-AP-1976} F. Magri,
%         An operator approach to Poisson brackets,
%         Ann. Phys., 99 (1976) 196-228.

\bibitem{ZNBC-PLA-2006} D.J. Zhang, T.K. Ning, J.B. Bi, D.Y. Chen,
         New symmetries for the Ablowitz-Ladik hierarchies,
         Phys. Lett. A, 359 (2006), 458-66.


\bibitem{LLTHC-JMP-2013} M.H. Li, C.Z. Li, K.L. Tian, J.S. He, Y. Cheng,
         Virasoro type algebraic structure hidden in the constrained discrete Kadomtsev-Petviashvili hierarchy,
         J. Math. Phys., 54 (2013) No.043512 (11pp).

\bibitem{LCH-MPLB-2013} M.H. Li, J.P. Cheng, J.S. He,
         The gauge transformation of the constrained semi-discrete KP hierarchy,
         Mod. Phys. Lett. B, 27 (2013), No.1350043 (13pp).


\bibitem{SS-IP-1991}J. Sidorenko, W. Strampp,
        Symmetry constraints of the KP hierarchy,
        Inverse Problems, 7 (1991) L37-43.

\bibitem{HI-IMRN-2000}L. Haine, P. Iliev,
        Commutative rings of difference operators and an adelic flag manifold,
        Int. Math. Res. Notices, 2000 (2000) 281-323.

\bibitem{LC-JPA-2010} S.W. Liu, Y. Cheng,
       Sato's B\"acklund transformations, additional symmetries and ASvM formula for the discrete KP hierarchy,
       J. Phys. A: Math. Theor., 43 (2010) No.135202 (11pp).





\end{thebibliography}
\end{document}